\documentclass[letterpaper, 10 pt, conference]{ieeeconf}  
\usepackage{algorithm}
\usepackage{algorithmicx}
\usepackage{stfloats}
\usepackage{graphicx}          
\usepackage{mathrsfs}         
\usepackage{amsmath}   
\usepackage{amsfonts}    
\usepackage{amssymb}   
\usepackage{booktabs}
\usepackage{threeparttable}
\usepackage{subfigure}
\usepackage{tabularx}
\usepackage{comment}
\usepackage[dvipsnames]{xcolor}
\usepackage{color}

\usepackage[font=small,skip=0pt]{caption}

\usepackage{enumerate}

\usepackage{amsthm}

\newtheorem{remark}{\textbf{Remark}}

\newtheorem{assumption}{\emph{\textbf{Assumption}}}
\newtheorem{theorem}{\textbf{Theorem}}
\newtheorem{lemma}{\textbf{Lemma}}

\newtheorem{problem}{Problem}
\DeclareMathOperator{\rank}{rank}
\DeclareMathOperator{\diag}{diag}

\newcommand*{\QEDA}{\hfill\ensuremath{\square}}   

\IEEEoverridecommandlockouts                   
\title{\LARGE \bf
	Learning Distributed Stabilizing Controllers for Multi-Agent Systems
}

\author{Gangshan Jing,~He Bai,~Jemin George,~Aranya Chakrabortty~and~Piyush K. Sharma
	\thanks{G.~Jing and A. Chkarabortty are with  North Carolina State University, Raleigh, NC 27695, USA.
		{\tt\small \{gjing, achakra2\}@ncsu.edu}}%
	\thanks{H.~Bai is with Oklahoma State University, Stillwater, OK 74078, USA.
		{\tt\small he.bai@okstate.edu}}%
	\thanks{J.~George and P.~Sharma are with the U.S. Army Research Laboratory, Adelphi, MD 20783, USA.
		{\tt\small \{jemin.george,piyush.k.sharma\}.civ@mail.mil}}%
}

\begin{document}

	\maketitle
	\thispagestyle{empty}
	\pagestyle{empty}

	\begin{abstract}
		We address the problem of model-free distributed stabilization of heterogeneous multi-agent systems using reinforcement learning (RL). Two algorithms are developed. The first algorithm solves a centralized linear quadratic regulator (LQR) problem without knowing any initial stabilizing gain in advance. The second algorithm builds upon the results of the first algorithm, and extends it to distributed stabilization of multi-agent systems with predefined interaction graphs. Rigorous proofs are provided to show that the proposed algorithms achieve guaranteed convergence if specific conditions hold. A simulation example is presented to demonstrate the theoretical results. 
	\end{abstract}


	\section{Introduction}
Reinforcement learning (RL) is a goal-oriented learning method where a system optimizes an intended policy according to a reward returned from its environment. Because of the generality of the approach, RL has found applications in diverse areas such as robotics~\cite{kober2013reinforcement}, communication~\cite{luong2019applications}, electric power systems \cite{sadamoto}, and defense-related military applications~\cite{barton2018reinforcement}. It has also been shown as a fantastic tool for solving optimal control problems, especially for linear quadratic regulator (LQR) design, when system dynamics are unknown \cite{Kiumarsi17}. A variety of formulations of RL has been proposed in the model-free LQR literature including methods such as adaptive dynamic programming (ADP) \cite{Jiang12}, Q-learning~\cite{Bradtke94,Lamperski20}, and zeroth-order optimization \cite{Fazel18}. Extensions of these centralized designs to distributed RL-based control have been reported in \cite{TCNS,Alemzadeh19,Li19}.

In this paper, we revisit the centralized LQR problem and the distributed stabilization of multi-agent networks with coupled dynamics using model-free RL. In the literature, almost all the RL-based LQR control methods require an initial stabilizing controller to start the learning algorithm even when the plant dynamics is known. In practice, however, due to the uncertainty of dynamics, knowing such initial stabilizing gains may not always be possible. Accordingly, the novelty of our work is to design RL algorithms for generating centralized and distributed stabilizing controllers without knowing explicit system dynamics. The problem of learning centralized stabilizing controllers has been recently addressed in \cite{Lamperski20,Persis19} for discrete-time systems, but the problem for continuous-time systems and distributed stabilization are yet to be studied. 

	
To resolve this issue, we propose  two  off-policy  RL algorithms based on the ADP technique for continuous-time linear systems with unknown dynamics. The first data-driven algorithm solves the centralized LQR problem without having any stabilizing gain as its initial guess. The second algorithm builds on the results of the first algorithm, and extends it to solving distributed stabilization of multi-agent systems with a predefined interaction graph. The fundamental idea is to introduce a damping parameter to the system, which is inspired by \cite{Feng20}. Our design, however, is quite different than \cite{Feng20} as we propose an {\it explicit} updating law for the damping parameter, followed by a rigorous proof of convergence. Moreover, unlike \cite{Feng20}, in our work the distributed stabilizing gain is learned based on a centralized stabilizing gain. The main results are illustrated using a simulation example.
	
This rest of the paper is structured as follows. In Section \ref{sec: LQR}, we introduce the conventional LQR problem and the damping formulation, followed by an off-policy RL algorithm based on ADP for deriving the optimal LQR without having a stabilizing gain in advance. In Section \ref{sec: dist}, we extend this design to a multi-agent stabilizing control problem. In Section \ref{sec: simulation}, a simulation example is presented to illustrate the effectiveness of the two algorithms. In Section \ref{sec: conclusion}, conclusions are drawn. The proof of Theorem \ref{th LQR} is presented in the appendix.
	
	\textbf{Notation}: Throughout the paper, given a matrix $X$, $X\succeq0$ implies that $X$ is positive semi-definite; $X\succ0$ implies that $X$ is positive definite; $\mathbb{S}_+$ is the set of positive definite matrices. For a symmetric matrix $X$, $\lambda_{\max}(X)$ and $\lambda_{\min}(X)$ denote the maximum and minimum eigenvalues of $X$, respectively. Let $\text{vec}(X)\in\mathbb{R}^{n^2}$ denote the vector stacking up columns of $X\in\mathbb{R}^{n\times n}$ from left to right. Let $\diag\{X_1,...,X_N\}$ denote a block diagonal matrix with $X_i$'s on the diagonal. $\mathbb{N}_+$ denotes the set of positive integers. The Euclidean 2-norm is denoted by $||\cdot||$.

	\section{Model-Free LQR without an Initial Stabilizing Controller}\label{sec: LQR}
	
	Consider the continuous-time LQR problem:
	\begin{equation}\label{CT LQR}
		\begin{split}
			\min_{K}~~~~&J(x(0),K)=\int_0^\infty (x^TQx+u^TRu)dt\\
			\text{s.t.}&~~~~\dot{x}(t)=Ax(t)+Bu(t), \\
			&~~~~u(t)=-Kx(t),
		\end{split}
	\end{equation}
	where $x(t)\in\mathbb{R}^n$ and $u(t)\in\mathbb{R}^m$ are the system state and the control input at time $t$, respectively, $A\in\mathbb{R}^{n\times n}$ and $B\in\mathbb{R}^{n\times m}$ are unknown. The variable $K$ is the control gain matrix, which is also referred to as the ``policy". $Q\succeq0$ and $R\succ0$ are both considered to be known. Such a setting is reasonable because $Q$ and $R$ determine the desired transient behavior of the system, thus are usually artificially designed.
	
	We make the following assumption to guarantee the existence and uniqueness of the solution to (\ref{CT LQR}):
	\begin{assumption}\label{as stob}
		$(A,B)$ is stabilizable, $(Q^{1/2},A)$ is observable.
	\end{assumption}
	The problem we aim to solve is formally stated as follows:
	\begin{problem}\label{continuous P}
		Given $Q$, $R$, and a collection of data $\mathcal{D}=\{x(t), u(t), t\in[0,Z\delta t]\}$ with $\delta t>0$ and $Z\in\mathbb{N}_+$ for the linear continuous-time system in (\ref{CT LQR}), find the optimal control gain $K$ for  the LQR problem (\ref{CT LQR}). 
	\end{problem}
	
	\subsection{The Damped LQR Problem and Kleinman's Algorithm}
	To solve (\ref{CT LQR}) without a stabilizing controller, similar to \cite{Feng20}, let us consider the damped problem:
	\begin{equation}\label{damped CT}
		\begin{split}
			\min_{K}~~J(\hat{x}(0),K,\alpha)&=\int_0^\infty e^{-2\alpha t}(\hat{x}^TQ\hat{x}+\hat{u}^TR\hat{u})dt\\
			\text{s.t.}~~~~\dot{\hat{x}}(t)&=A\hat{x}(t)+B\hat{u}(t),\\
			~~~~\hat{u}(t)&=-K\hat{x}(t),
		\end{split}
	\end{equation}
	where $\alpha\geq0$ is the damping parameter. By setting $x(t)=e^{-\alpha t}\hat{x}(t)$ and $u(t)=e^{-\alpha t}\hat{u}(t)$, problem (\ref{damped CT})  is equivalent to the following problem in the sense that with the same control policy and $x(0)=\hat{x}(0)$, they have equivalent cost function values:
	\begin{equation}\label{damped dynamics}
		\begin{split}
			\min_{K}~~~~J(x(0),K,\alpha)&=\int_0^\infty (x^TQx+u^TRu)dt\\
			\text{s.t.}~~~~\dot{x}(t)&=(A-\alpha I)x(t)+Bu(t),\\
			u(t)&=-Kx(t).
		\end{split}
	\end{equation}

	When $A$ and $B$ are known, given $\alpha$ such that $(A-\alpha I_n,B)$ is stabilizable, it is well-known that the solution to (\ref{damped dynamics}) is
	\begin{equation}
		K^*=R^{-1}B^TP,
	\end{equation}
	where $P$ is the solution to the following algebraic Riccati equation:
	\begin{equation}\label{CT Riccati}
		(A-\alpha I)^TP+P(A-\alpha I)+Q-PBR^{-1}B^TP=0.
	\end{equation}
	Given a stabilizing gain $K_0$, Kleinman's algorithm \cite{Kleinman68} shows that equation (\ref{CT Riccati}) can be obtained by alternatively solving for $P_k$ and $K_k$ such that
	\begin{equation}\label{CT Lyapunov}
		(A-\alpha I-BK_k)^TP_k+P_k(A-\alpha I-BK_k)+Q+K_k^TRK_k=0,
	\end{equation}
	\begin{equation}\label{CT K}
		K_{k+1}=R^{-1}B^TP_k, k=0,1,2,....
	\end{equation}
	The updating mechanism (\ref{CT Lyapunov})-(\ref{CT K}) can be viewed as {\it policy iteration}. Let $$V_k(x(0))=x^T(0)P_kx(0)$$ be the cost to go function evaluating the cost in (\ref{damped dynamics}) with $u(t)=K_{k+1}x(t)$ as the controller. A good property of the iteration algorithm is the monotonic decreasing of $P_k$ during the evolution:
	\begin{equation}
		P_{k}\preceq P_{k-1}.
	\end{equation}
	As a result, $K_{k+1}$ is still stabilizing and actually yields a better performance. Solving for $P_k$ from (\ref{CT Lyapunov}) is {\it policy evaluation}, and the updating of $K_{k+1}$ in (\ref{CT K}) is {\it policy improvement}.
	
	\subsection{Model-Free RL for the Damped LQR}
	
	The model-free RL algorithm is to achieve policy iteration without having the explicit information of the dynamics model, while instead, based on control input and state data only.
	
	Let $A_k=A-BK_k$.  Based on system dynamics in (\ref{damped dynamics}), the updating laws of $P_k$ and $K_{k+1}$, we study the value function change for a given policy $K_{k+1}$ during the evolution of the original system in (\ref{CT LQR}) in time interval $[t,t+\delta t]$:
	\begin{equation}\label{integral}
		\begin{split}
			&x^T(t+\delta t)P_kx(t+\delta t)-x^T(t)P_kx(t)\\
			&=\int_t^{t+\delta t}\left[x^T(A_k^TP_k+P_kA_k)x
			+2(u+K_kx)^TB^TP_kx\right]ds\\
			&=\int_t^{t+\delta t}x^T(2\alpha P_k-Q-K_k^TRK_k)xds+\\
			&~~~~~~~~~2\int_t^{t+\delta t}(u+K_kx)^TRK_{k+1}xds,
		\end{split}
	\end{equation}
	where the last equality is obtained based on (\ref{CT Lyapunov}) and (\ref{CT K}). 
	
	Define the mapping $\mu(\cdot): \mathbb{R}^{n}\rightarrow\mathbb{R}^{n(n+1)/2}$ on vector $y=(y_1,...,y_n)^T\in\mathbb{R}^n$ and the mapping $\nu(\cdot): \mathbb{R}^{n\times n}\rightarrow\mathbb{R}^{n(n+1)/2}$ on matrix $X=[X_{ij}]$ such that
	\begin{equation}
		\mu(y)=(y_1^2, y_1y_2, ..., y_1y_n, y_2^2, y_2y_3, ..., y_{n-1}y_n, y_n^2)^T,
	\end{equation}
	\begin{equation}
		\nu(X)=(X_{11},2X_{12},..., 2X_{1n},X_{22}, 2X_{23},..., X_{n-1,n}, X_{nn})^T.
	\end{equation}
	We also let
	\begin{equation}
		\delta_{xx}=\left[\mu(x)|_{0}^{\delta t},...,\mu(x)|_{(Z-1)\delta_t}^{Z\delta t}\right]^T\in\mathbb{R}^{Z\times n(n+1)/2},
	\end{equation}
	\begin{equation}
		I_{x}=\left[\int_0^{\delta t}\mu(x)ds,...,\int_{(Z-1)\delta t}^{Z\delta t}\mu(x)ds\right]^T\in\mathbb{R}^{Z\times n(n+1)/2},
	\end{equation}
	\begin{equation}
		I_{xx}=\left[\int_0^{\delta t}(x\otimes x)ds,...,\int_{(Z-1)\delta t}^{Z\delta t}(x\otimes x)ds\right]^T\in\mathbb{R}^{Z\times n^2},
	\end{equation}
	\begin{equation}
		I_{xu}=\left[\int_0^{\delta t}(x\otimes u)ds,...,\int_{(Z-1)\delta t}^{Z\delta t}(x\otimes u)ds\right]^T\in\mathbb{R}^{Z\times mn}.
	\end{equation}
	Then equation (\ref{integral}) is equivalent to
	\begin{equation}\label{CT LS}
		\Theta_k(\alpha)\begin{pmatrix}
			\nu(P_k)\\ \text{vec}(K_{k+1})
		\end{pmatrix}=\Xi_k,,
	\end{equation}
	where $\Theta_k(\alpha)=(\delta_{xx}-2\alpha I_{x},-2I_{xu}(I_n\otimes R)-2I_{xx}(I_n\otimes K_k^TR))\in\mathbb{R}^{Z\times [n(n+1)/2+mn]}$, $\Xi_k=-I_{x}\nu(Q+K_k^TRK_k)\in\mathbb{R}^{Z}$.
	
	\begin{assumption}\label{as rank}
		$\rank((I_x,I_{xu}))=\frac{n(n+1)}{2}+mn$.
	\end{assumption}

	If the data set $\mathcal{D}$ satisfies Assumption~\ref{as rank}, then solving the least squares problem (\ref{CT LS}) is equivalent to updating $P_k$ and $K_{k+1}$ according to (\ref{CT Lyapunov}) and (\ref{CT K}). Noe that the data set $\mathcal{D}$ can be independent of the control policy $K_k$ to be updated. Hence, the RL algorithm based on solving (\ref{CT LS}) is {\it off-policy}.

	\subsection{RL without Stabilizing Initialization}
	It has been shown in \cite{Jiang12} that under Assumptions \ref{as stob} and \ref{as rank}, with a stabilizing gain at hand, the optimal controller $K^*$ for problem (\ref{CT LQR}) can be obtained by repeatedly solving (\ref{CT LS}) with $\alpha=0$. In this subsection, we propose a model-free approach for seeking $K^*$ without having a stabilizing gain in advance.	
	
	Observe that if $\alpha$ is sufficiently large, $(A-\alpha I,B)$ is always stabilizable. In this work, we will start with a sufficiently large $\alpha_0$ such that $K(0)=\mathbf{0}_{m\times n}$ is stabilizing, i.e., $A-\alpha_0 I_n$ is Hurwitz. Then we decrease $\alpha_k$ until $(A-\alpha_k I,B)$ is not stabilizable. Under Assumption \ref{as stob}, $\alpha_k$ converges to 0 as $k\rightarrow\infty$. Different from \cite{Feng20}, we will propose a model-free approach and specify the updating laws for $\alpha$ and $K$.
	
	The fundamental idea is to update $\alpha_{k+1}$ and $K_{k+1}$ alternatively. During each iteration, we firstly find the minimum $\alpha_{k+1}$ such that $A-BK_k-\alpha_{k+1}I_n$ is stable, then update $K_{k+1}$ such that $A-BK_{k+1}-\alpha_{k+1}I_n$ is stable. The minimum $\alpha_{k+1}$ can be obtained by solving the following optimization:
	\begin{equation}\label{solve alpha}
		\begin{split}
			&\min_{\alpha_{k+1},P_k,K_{k+1}} \alpha_{k+1}\\
		\text{s.t.}~~~~	&\alpha_{k+1},P_k,K_{k+1} ~\text{satisfies}~ (\ref{CT LS}),\\
			&||P_k-P_{k-1}||<\sigma,\\
			&P_k\succ0,\alpha_{k+1}\geq0,
		\end{split}
	\end{equation}
	where $\sigma>0$ is a predefined constant. Here $\sigma$ can be any positive constant. Note that if $\sigma$ is very large, the solution $\alpha_{k+1}$ to (\ref{solve alpha}) may render $A-\alpha_{k+1}I-BK_k$ very close to an unstable matrix, thereby leading to a matrix $P_k$ with very large eigenvalues. Hence, the second constraint in (\ref{solve alpha}) is to restrict $||P_k||$.
	
	The solution to (\ref{solve alpha}) always makes $A-\alpha_{k+1} I_n-BK_k$ stable because the corresponding $P_k$, which is actually the solution to (\ref{CT Lyapunov}) for $\alpha_{k+1}$, is positive definite. Moreover, the corresponding $K_{k+1}$ must be stabilizing for $A-\alpha_{k+1}I_n$ because it is actually the policy improvement result (\ref{CT K}).
	
	Observe that both $\alpha_{k+1}$ and $P_k$ are variables in (\ref{solve alpha}), making the optimization (\ref{solve alpha}) nonlinear. In fact, in our algorithm, we do not need the exactly optimal solution to (\ref{solve alpha}). Hence, we will solve for an approximate optimal solution by using Algorithm \ref{alg:alpha}. The RL algorithm for solving Problem \ref{continuous P} is given in Algorithm \ref{alg:CTARL}, which contains Algorithm \ref{alg:alpha} as the policy improvement step.

	\begin{algorithm}[htbp]
		\small
		\caption{Alternating RL Algorithm for Problem \ref{CT LQR}}\label{alg:CTARL}
		\textbf{Input}: $Q$, $R$, $\mathcal{D}=\{x(t), u(t), t\in[0, Z\delta t]\}$, $\eta=\alpha_0/S$ with $S\in\mathbb{N}_+$,  $\alpha_0>0$, threshold $\epsilon>0$.\\
		\textbf{Output}: $K^*$.
		\begin{itemize}
			\item[1.] Set $k\leftarrow0$, $K_0\leftarrow\mathbf{0}_{m\times n}$, and $P_{-1}\leftarrow\mathbf{0}_{n\times n}$. Compute  $\delta_{xx}$, $I_{xx}$, $I_{xu}$ and $\Xi_k$.
			\item[2.] \textbf{while} $\alpha_k\geq\eta$ 
			
			Obtain $\alpha_{k+1}$, $P_k$ and $K_{k+1}$  by implementing Algorithm \ref{alg:alpha}, set $k\leftarrow k+1$. 
			
			\textbf{end while}
			\item[3.] \textbf{while}
			$\frac{||K_{k+1}-K_k||}{||K_k||}\geq\epsilon$
			
			Set $\alpha_{k+1}=0$. Obtain $P_k$ and $K_{k+1}$  by solving (\ref{CT LS}). 
			
			Set $k\leftarrow k+1$. 
			
			\textbf{end while}
			\item[4.] Set $K^*=K_{k+1}$.
		\end{itemize}
	\end{algorithm}	
	
	\begin{algorithm}[htbp]
		\small
		\caption{Updating $\alpha_{k+1}$, $P_k$, $K_{k+1}$ for Problem \ref{continuous P}}\label{alg:alpha}
		\textbf{Input}: $\delta_{xx}$, $I_{xx}$, $I_{xu}$, $\Xi_k$,  $P=P_{k-1}$, $K=K_k$, $\alpha_k$, $\sigma>0$, step size $\eta=\alpha_0/S$ with $S\in\mathbb{N}_+$.\\
		\textbf{Output}: $\alpha_{k+1}$, $P_k$, and $K_{k+1}$.
		\begin{itemize}
			\item[1.] Set $l=0$.
			
			\item[2.] \textbf{while} $\alpha_{k}\geq\eta$
			
			\item[3.] Set $\alpha_{k+1}\leftarrow\alpha_k-\eta$. Compute $\Theta_k(\alpha_{k+1})$. 
			\item[4.] Obtain $P_k$ and $K_{k+1}$ by solving (\ref{CT LS}).
			\item[5.] \textbf{if} $P_k\succ0$ and $||P_k-P_{k-1}||<\sigma$ 
			
			Set $\alpha_k\leftarrow\alpha_{k+1}$, $P\leftarrow P_k$, $K\leftarrow K_{k+1}$, $l\leftarrow l+1$
			
			\textbf{else}
			
			Go to step 7.
			
			\textbf{end if}
			
			\item[6.] \textbf{end while}
			\item[7.] \textbf{if} $l=0$
			
			Set $\alpha_{k+1}\leftarrow\alpha_{k}$, obtain $P_k$ and $K_{k+1}$ by solving (\ref{CT LS}). 
			
			\textbf{else} 
			
			Set $\alpha_{k+1}\leftarrow\alpha_{k}$, $P_k\leftarrow P$, $K_{k+1}\leftarrow K$.
			
			\textbf{end if}
		\end{itemize}
	\end{algorithm}		
	
	\begin{remark}\label{re alpha=0}
	Since $\eta=\alpha_0/S$ in Algorithm~\ref{alg:alpha}, theoretically $\alpha_k$ converges to 0 if it is decreased by $\eta$ for $S$ times. In numerical simulations, there may be errors when computing $\eta$. As a result, $\alpha_{k}$ may be a small scalar when $\alpha_k<\eta$. However, this error can be avoided by first giving $\eta$, and then selecting a large enough $S$ such that $\alpha_0=S\eta$ renders $A-\alpha_0I_n$ stable. Then for any step $k$, $\alpha_k=0$ if $\alpha_k<\eta$.
	\end{remark}

	\begin{theorem}\label{th LQR}
		Suppose $Q\succ0$. Under Assumptions \ref{as stob}, \ref{as rank}, by implementing Algorithm \ref{alg:CTARL}, there exist $\eta>0$ and $S_0>0$, such that if $S\geq S_0$, then $K_{k}$ converges to the optimal control gain as $k\rightarrow\infty$.
	\end{theorem}
	
	\begin{remark}
		In practice, matrix $Q$ in the LQR problem may not satisfy the positive definiteness assumption. In this scenario, we can firstly obtain a stabilizing gain by implementing Algorithm \ref{alg:CTARL} with an artificially designed positive definite $Q'$, then run a conventional ADP algorithm, e.g.,~\cite{Jiang12}, to learn the optimal controller for the original LQR problem.
	\end{remark}

	\section{Learning a Distributed Stabilizing Controller}\label{sec: dist}
	
	Suppose that the system in (\ref{CT LQR}) is a heterogeneous multi-agent system with $N$ agents, where  the system equation for each agent is written as
	\begin{equation}\label{MAS}
		\dot{x}_i=A_{ii}x_i+\sum_{j\in\mathcal{N}_i}A_{ij}x_j+B_iu_i, ~~i=1,...,N,
	\end{equation}
	where $x_i\in\mathbb{R}^{\bar{n}_i}$ and $u_i\in\mathbb{R}^{\bar{m}_i}$ are the state and control input of agent $i$, respectively,  $A_{ij}\in\mathbb{R}^{\bar{n}_i\times\bar{n}_i}$, and $\mathcal{N}_i$ is the set of agents whose dynamics are coupled with agent $i$. Note that different agents may have different dimensions for states and control inputs. 
	
	Define $n=\sum_{i=1}^N\bar{n}_i$ and $m=\sum_{i=1}^N\bar{m}_i$.	Let $x=(x_1^T,...,x_N^T)^T$, $u=(u_1^T,...,u_N^T)^T$,  $A=[A_{ij}]\in\mathbb{R}^{n\times n}$ with $A_{ij}=\mathbf{0}_{\bar{n}_i\times\bar{n}_j}$ for $(i,j)\notin\mathcal{E}$ and $B=\diag\{B_1,...,B_N\}\in\mathbb{R}^{n\times m}$. Then we are able to write the compact system dynamics as
	\begin{equation}
		\dot{x}=Ax+Bu.
	\end{equation}
	
	Let an undirected graph $\mathcal{G}=(\mathcal{V},\mathcal{E})$ denote the communication graph interpreting the desired interaction relationship among agents. Here $\mathcal{V}=\{1,...,N\}$ is the set of agents,  $\mathcal{E}\subset\mathcal{V}\times\mathcal{V}$ is the set of edges specifying those pairs of agents that are able to utilize information of each other. Our goal is to find a structured stabilizing gain matrix $K\in\mathbb{R}^{m\times n}$ for the multi-agent system (\ref{MAS})  such that the controller of agent $i$ determined by $u=-Kx$ involves the state of agent $j$ if and only if $(i,j)\in\mathcal{E}$. In other words, the goal is to find $K$ such that $A-BK$ is stable, and $K\in\mathcal{S}_K(\mathcal{G})$, where the set $\mathcal{S}_K(\mathcal{G})$ is defined as follows:
	\begin{equation}\label{SKG}
		\mathcal{S}_K(\mathcal{G})=\{K\in\mathbb{R}^{m\times n}: K(i,j)=\mathbf{0}_{\bar{m}_i\times\bar{n}_j}~\text{if}~(i,j)\notin\mathcal{E}\},
	\end{equation}
	where $K(i,j)$ is a submatrix of $K$ composed of elements from the $(\sum_{k=1}^{i-1}\bar{m}_k+1)$-th to $\sum_{k=1}^i\bar{m}_k$-th rows and from $(\sum_{l=1}^{j-1}\bar{n}_l+1)$-th to $\sum_{l=1}^j\bar{n}_l$ columns of $K$.
	
	The problem we aim to solve in this section is formally stated as follows:	
	\begin{problem}\label{continuous DP}
		Given a collection of data $\mathcal{D}=\{x(t), u(t),t\in[0,Z\delta t]\}$ with $\delta t>0$ and $Z\in\mathbb{N}_+$ for the linear continuous-time system in (\ref{MAS}), find a control gain matrix $K_d\in\mathcal{S}_K(\mathcal{G})$ such that $A-BK_d$ is stable. 
	\end{problem}
	
	Under Assumptions \ref{as stob} and \ref{as rank}, by artificially designing a cost function in (\ref{CT LQR}) with $Q=I_n$ and $R=I_m$, and implementing Algorithm \ref{alg:CTARL}, we can obtain a stabilizing gain $K_s$ when $\alpha_k$ converges to 0. Next we  show how to obtain a distributed stabilizing gain $K_d$ based on an available stabilizing gain $K_s$.
	
	Similar to (\ref{SKG}), define
	\begin{equation}
	    \mathcal{S}_P(\mathcal{G})=\{P\in\mathbb{R}^{n\times n}: P(i,j)=\mathbf{0}_{\bar{n}_i\times\bar{n}_j}~\text{if}~(i,j)\in\mathcal{E}\},
	\end{equation}
	\begin{equation}
		\mathcal{S}_R=\{R\in\mathbb{R}^{m\times m}: R(i,j)=\mathbf{0}_{\bar{m}_i\times \bar{m}_j}~ \text{if}~ i\neq j\}.
	\end{equation}
	
	The following lemma is the theoretical foundation of our approach.
	
	\begin{lemma}\label{le APd}
		There exists a distributed stabilizing gain $K_d\in\mathcal{S}_K$ for (\ref{MAS}) if the following problem is feasible:
		\begin{equation}\label{APd}
			\begin{split}
				&\text{find} ~~P\\
				\text{s.t.}~~	(A-BK_s)^T&P+P(A-BK_s)\prec0,\\
				P\in&\mathcal{S}_P(\mathcal{G}),~~P\succ0.
			\end{split}	
		\end{equation}
	\end{lemma}	
	\begin{proof}
		Suppose that $P_d$ is a solution to (\ref{APd}). Denote
		$$D=-\left[(A-BK_s)^TP_d+P_d(A-BK_s)\right]\succ0.$$
		For any $R'\succ0$ such that $R'\in\mathcal{S}_R$, choose $s>0$ such that $sD\succeq K_s^TR'K_s$. Then $sP_d$ is still a solution to (\ref{APd}). Let $Q'=sD-K_s^TR'K_s$. Let $K_d=s{R'}^{-1}B^TP_d$. We can view $K_d$ as a policy improvement from the current policy $K_s$ for the following cost function:
		\begin{equation}
			J_k(x(0),u)=\int_0^\infty (x^TQ'x+u^TR'u)dt.
		\end{equation}
		Therefore, $K_d$ is stabilizing as well. Moreover, since $B$ and $R'$ are block-diagonal and $P\in\mathcal{S}_P(\mathcal{G})$, we have $K_d\in\mathcal{S}_K$.
	\end{proof}

	In a model-free way, similar to the last section, (\ref{APd}) can be transformed to the data-based form:
	\begin{equation}\label{Stab LS}
		\begin{split}
			\text{find}~~P_d, D, E&\\
			\text{s.t.}~~~~~	\Theta(K_s)\begin{pmatrix}
				\nu(P_d)\\
				\text{vec}(E)
			\end{pmatrix}&=-I_x\nu(D),\\
			D\succ0, P_d\in\mathcal{S}_P,&P_d\succ0,
		\end{split}
	\end{equation}
	where $\Theta(K_s)=(\delta_{xx},-2I_{xu}(I_n\otimes R)-2I_{xx}(I_n\otimes K_s^TR))\in\mathbb{R}^{Z\times [n(n+1)/2+mn]}$. According to (\ref{CT LS}), we know that for any solution $\{P_d,D,E\}$ to (\ref{Stab LS}), it actually holds that $E=R^{-1}B^TP_d$. Here matrix $R\succ0$ is artificially designed and will be used to compute $\Theta(K_s)$.
	
	Next we transform (\ref{Stab LS}) to a semi-definite program (SDP), by converting positive definite constraints to positive semi-definite constraints. Note that for any solution $\{P_d, D, E\}$ to \eqref{Stab LS}, $\{sP_d,sD,sE\}$ for any $s>0$ is still a solution. That is, once (\ref{APd}) is feasible, there must exist a solution $\{P_d, D, E\}$ to (\ref{APd})  such that $D\succeq cI_n$ and $P_d\succeq cI_n$ for any $c>0$. In practice, there may be a desired range for the traces of the solution matrices. Without loss of generality, we aim to solve the following linear SDP:
	\begin{equation}\label{Stab SDP}
		\begin{split}
			\min_{P_d,D,E}~~\text{trace}(P_d)&\\
			\text{s.t.}~~~~~	\Theta(K_s)\begin{pmatrix}
				\nu(P_d)\\
				\text{vec}(E)
			\end{pmatrix}&=-I_x\nu(D),\\
			D\succeq cI_n,~ P_d\in\mathcal{S}_P(\mathcal{G}),~& P_d\succeq cI_n,
		\end{split}
	\end{equation}
	where $c>0$ is artificially given depending on the requirement in practice.
	
	\begin{lemma}\label{le equi}
		Under Assumption \ref{as rank}, (\ref{APd}) is feasible if and only if (\ref{Stab SDP}) is feasible.
	\end{lemma}
	\begin{proof}
		Sufficiency. By following similar lines to the proof of \cite[Lemma 6]{Jiang12}, the validity of Assumption \ref{as rank} implies that a solution $\{P_d, D, E\}$ to (\ref{Stab SDP}) must satisfy
		\begin{equation}\label{DkAPk}
			-D=(A-BK_s)^TP_d+P_d(A-BK_s).
		\end{equation}
		Therefore, $P_d$ must be a solution to (\ref{APd}).
		
		The necessity can be proved by noting that for any solution $P_d$ to (\ref{APd}), $sP_d$ with any $s>0$ is still a solution.
	\end{proof}	
	Now we present Algorithm \ref{alg:CTDARL} as the algorithm for learning a distributed stabilizing controller for (\ref{MAS}).
	\begin{algorithm}[htbp]
		\small
		\caption{RL Algorithm for Problem \ref{continuous DP}}\label{alg:CTDARL}
		\textbf{Input}: $\mathcal{D}=\{x(t), u(t), t\in[0,Z\delta t]\}$, $R'\in\mathbb{S}_+\cap\mathcal{S}_R$, $\alpha_0>0$, $c>0$.\\
		\textbf{Output}: $K_d$.
		\begin{itemize}
			\item[1.] Set $Q=I_n$, $R=I_m$. Implement Algorithm \ref{alg:CTARL} until $\alpha_k=0$, obtain a stabilizing gain $K_s$.
			\item[2.] Solve SDP (\ref{Stab SDP}), obtain solution $\{P_d, D, E\}$. Let $K_d=s{R'}^{-1}E$, where 
			\begin{equation}\label{s}
				s\geq\frac{\lambda_{\max}(K_s^TR'K_s)}{\lambda_{\min}(D)}.
			\end{equation}
		\end{itemize}
	\end{algorithm}	
	
	The block-diagonal input matrix $R'$ in Algorithm \ref{alg:CTDARL} can be designed depending on the requirement in practice, which affects the obtained control gain. The simplest way to design $R'$ is setting $R'=I_m$. The coefficient $s$ in Algorithm \ref{alg:CTDARL} makes the resulting $K_d$ stabilizing as long as it satisfies (\ref{s}).
	
	\begin{theorem}\label{th MAS}
		Under Assumption \ref{as rank}, if (\ref{APd}) is feasible, then Algorithm \ref{alg:CTDARL} solves Problem \ref{continuous DP}.
	\end{theorem}
	\begin{proof}
		The formula for $s$ in (\ref{s}) ensures $sD\succeq K_s^TR'K_s$. By following similar lines to the proof of \cite[Lemma 6]{Jiang12}, a solution $\{P_d, D, E\}$ to (\ref{Stab SDP})  satisfies $E=R^{-1}B^TP_d=B^TP_d$. Together with the proof of Lemma \ref{le APd}, we obtain that $K_d=s{R'}^{-1}E=s{R'}^{-1}B^TP_d\in\mathcal{S}_K$ is stabilizing.
	\end{proof}

	\begin{remark}
Both Algorithm \ref{alg:CTARL} and	Algorithm \ref{alg:CTDARL} are centralized because both (\ref{CT LS}) and (\ref{Stab SDP}) involve the overall control input $u$ and state $x$ for the system. However, the controller obtained by implementing Algorithm \ref{alg:CTDARL} is distributed, because the controller of each agent only involves its neighbors' state information.
\end{remark}
	
	\section{Simulations}\label{sec: simulation}	
	
	Consider a multi-agent system (\ref{MAS}) with 3 agents. The overall system matrix and control input matrix are given in (\ref{A B}).
	\begin{figure*}
		\begin{equation}\label{A B}
			A=\begin{pmatrix}
				0.48&	0.23	&0.89&	0.86	&0&	0\\
				0.12&	0.07&	0.16&	0.73&	0&	0\\
				0.64&	0.03&	0.57&	0.71&	0.65&	0.30\\
				0.47&	0.16&	0.62&	0.25&	0.13&	0.67\\
				0&	0&	0.40&	0.95&	0.11&	0.63\\
				0&	0&	0.14&	0.69&	0.90&	0.08
			\end{pmatrix}, B=\begin{pmatrix}
				0.37&	0&	0\\
				0.92&	0&	0\\
				0&	0.09&	0\\
				0&	0.52&	0\\
				0&	0&	0.91\\
				0&	0&	0.31
			\end{pmatrix}.
		\end{equation}
	\end{figure*}
	The interaction topology is considered to be consistent with inter-agent dynamics coupling relationship reflected by matrix $A$ shown in Fig. \ref{fig topology}. By setting $\alpha_0=2.46$, $\sigma=100$, $\eta=0.001$, $Q=I_6$, $R=I_3$, the result of implementing Algorithm \ref{alg:CTARL} is shown in Fig. \ref{fig LQR}. Observe that as $\alpha_{k+1}$ evolves, $\lambda_{\max}(P_k)$ may be increased or decreased. Once $\alpha_{k+1}$ converges to 0, $\lambda_{\max}(P_k)$ keeps decreasing, and ultimately converges to $\lambda_{\max}(P^*)$, where $P^*$ is the optimal cost to go matrix for (\ref{CT LQR}). Moreover, $K_k$ asymptotically converges to the optimal control gain, which is as follows:
	\begin{equation}\label{Ks}
		K^*=\begin{pmatrix}
			3.51&	0.86&	3.82&	2.53&	0.62&	0.23\\
			4.36&	0.05&	5.59&	4.34&	1.63&	1.32\\
			1.75&	-0.01&	3.17&	3.09&	2.45&	2.18
		\end{pmatrix}.
	\end{equation}
	
	By taking $K^*$ as the stabilizing gain matrix $K_s$ and implementing Algorithm \ref{alg:CTDARL} with $c=100$ and $R'=I_3$, the obtained distributed stabilizing control gain is
	\begin{equation}
		K_d=\begin{pmatrix}
			139.55&	102.25&	73.54&	33.76&	0&	0\\
			174.73&	-44&   165.15&	142.76&	4.52&	3.97\\
			0&	0&	91.70&	-5.61&	179.35&	91.79
		\end{pmatrix}.
	\end{equation}
	The structure of $K_d$ shows that the controllers of agents 1 and 3 do not involve state information of each other, thus are distributed controllers.
	
	\begin{figure}
		\centering
		\includegraphics[width=6cm]{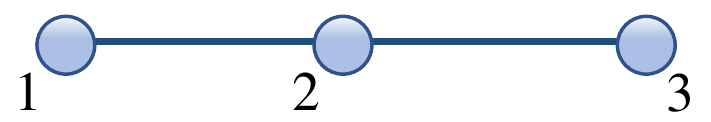}
		\caption{The desired interaction graph.} \label{fig topology}
	\end{figure}
	
	\begin{figure}
		\centering
		\includegraphics[width=9cm]{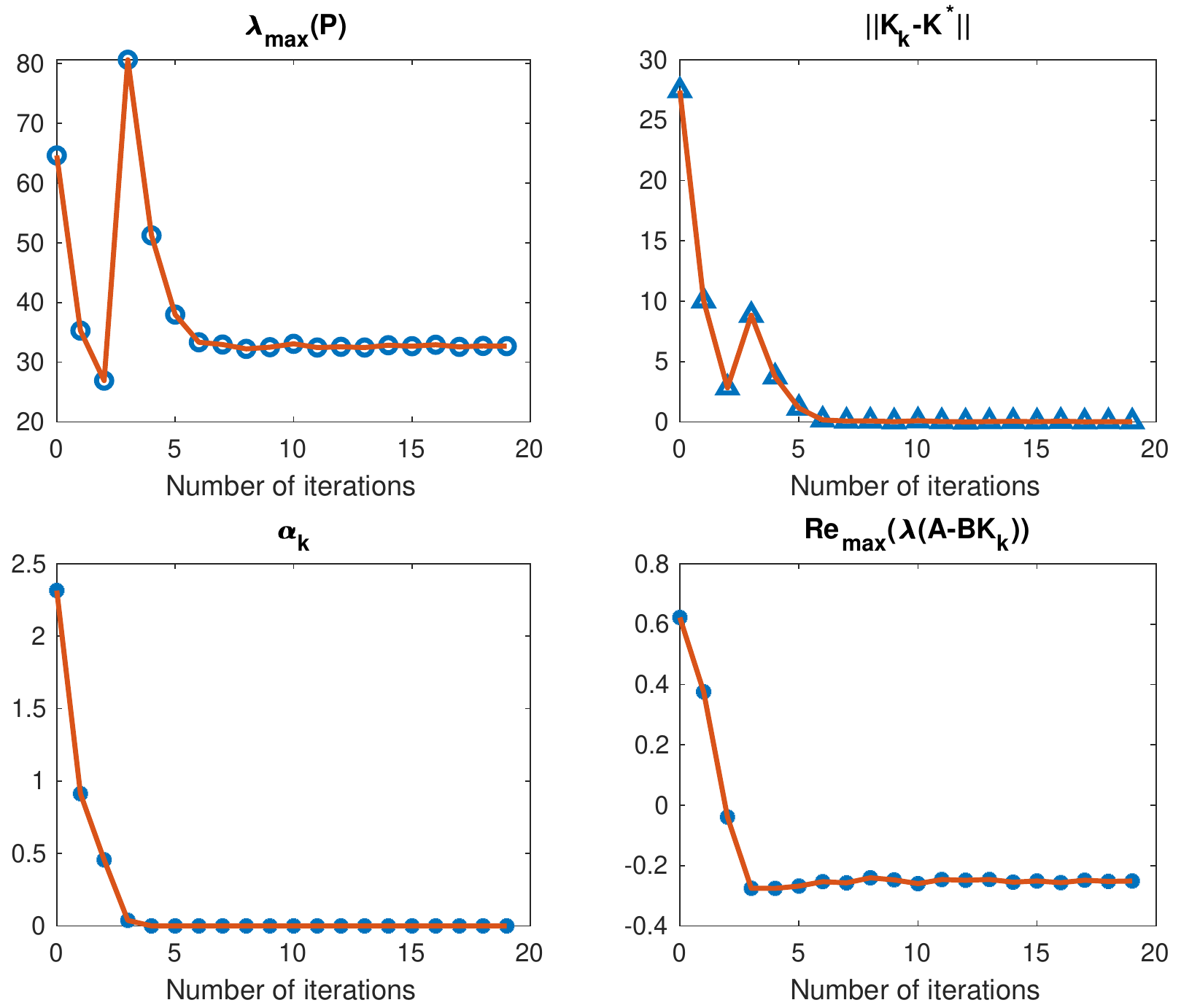}
		\caption{The results by implementing Algorithm \ref{alg:CTARL}.} \label{fig LQR}
	\end{figure}

	
	\section{Conclusions}\label{sec: conclusion}
	
	We have proposed two off-policy model-free RL algorithms for the optimal control of general linear systems and for the distributed stabilization of linear multi-agent systems. By introducing a damping parameter, the RL-based LQR control  methods can be generalized to scenarios where an initial stabilizing gain is not available. Once a centralized stabilizing gain is learned, a distributed stabilizing gain with a desired distributed structure can be learned for a multi-agent system. Though both learning algorithms are currently centralized, we hope to develop a distributed learning scheme to construct the distributed stabilizing controller.

	\section{Appendix: Proof of Theorem \ref{th LQR}}\label{sec:appendix}
	
	Before entering into the proof, we first present a supporting lemma. Let $K^*(\alpha)$ be the optimal control gain for (\ref{damped CT}). The following lemma shows a robustness property of $K^*(\alpha)$.
	\begin{lemma}\label{le delta1}
		Suppose that $Q\succ0$. Given $\alpha_0>0$, there exists $\delta_1>0$, such that for any $\alpha\in[0,\alpha_0]$, if $||K-K^*(\alpha)||<\delta_1$, then $A-\alpha I-BK$ is stable.
	\end{lemma}
	\begin{proof}
		For simplicity, we denote $A_\alpha=A-\alpha I_n$, $K^*_\alpha=K^*(\alpha)$. From the definition of $K^*(\alpha)$, there exists $P_\alpha\succ0$, such that
		\begin{equation}
			(A_\alpha-BK^*_\alpha)^TP_\alpha+P_\alpha(A_\alpha-BK^*_\alpha)+Q+{K^*_\alpha}^TRK^*_\alpha=0,
		\end{equation}	
		Let $\Delta K=K^*-K$, we have
		\begin{equation}\label{K stabilizing}
			\begin{split}
				&(A_\alpha-BK)^TP_\alpha+P_\alpha(A_\alpha-BK)\\
				&=(A_\alpha-BK^*_\alpha)^TP_\alpha+P_\alpha(A_\alpha-BK^*_\alpha)+\Delta K^TP_\alpha+P_\alpha\Delta K\\
				&=-Q-{K^*_\alpha}^TRK^*_\alpha+\Delta K^TP_\alpha+P_\alpha\Delta K.
			\end{split}
		\end{equation}
		From \cite[Corollary 2]{Feng20}, $K^*_\alpha$ is continuous on $\alpha$, which implies that $P_\alpha$ is continuous on $\alpha$. As a result, there exists $\lambda_P$ such that
		$$\lambda_{\max}(P_\alpha)\leq \lambda_P, ~~~\alpha\in[0,\alpha_0].$$
		Equation (\ref{K stabilizing}) means that $A_\alpha-BK$ is stable as long as $$\lambda_{\max}(\Delta K^TP_\alpha+P_\alpha\Delta K)<\lambda_{\min}(Q),$$
		which holds if
		\begin{equation}
			||\Delta K||<\frac{\lambda_{\min}(Q)}{2\lambda_P}\triangleq\delta_1,
		\end{equation}
		here $||\cdot||$ is the induced 2-norm.
	\end{proof}

	{\it Proof of Theorem \ref{th LQR}.} The updating mechanism of $\alpha_{k+1}$ in Algorithm \ref{alg:alpha} implies that $\alpha_k$ is nonincreasing in $k$. Together with $\alpha\geq0$, we have $\lim_{k\rightarrow\infty}\alpha_k=\alpha^*$ for some $\alpha^*\geq0$. Note that when $\alpha_k\geq\eta$, $\alpha_{k+1}$ is always obtained by implementing Algorithm \ref{alg:alpha}. It can be observed that the outputs of Algorithm \ref{alg:alpha} always ensure stability of $A_{k}-\alpha_{k}I_n$. However, if $A_{k+1}-(\alpha_k-\eta) I_n$ is not stable, $\alpha_{k+1}$ remains to be $\alpha_k$ while $K_{k+1}$ keeps updating. 
	
	We now prove that $\eta$ can be chosen such that $\alpha_k$ converges to $\alpha^*<\eta$ as $k\rightarrow\infty$. It suffices to show that for any step $k$, if $\alpha_k\geq\eta$, then there always exists $k'>k$ such that $A-(\alpha_k-\eta)I_n-BK_{k'}$ is stable, and based on $\alpha_k-\eta$ and $K_{k'}$, the updated $P_{k'}$ satisfies $||P_{k'}-P_{k'-1}||<\sigma$.
	
	Let $\alpha'=\alpha_k-\eta$. Recall that $K^*_\alpha$ is continuous on $\alpha$, so $K^*_\alpha$ is uniformly continuous for $\alpha\in[0,\alpha_0]$. Then there exists $\delta_2>0$, such that once $|\alpha-\alpha'|<\delta_2$, there must hold that $$||K^*_{\alpha}-K^*_{\alpha'}||<\delta_1/2,$$
	where $\delta_1$ is specified in Lemma \ref{le delta1}.
	
	As shown in \cite{Jiang12}, under Assumption \ref{as rank}, when $\alpha_k$ remains unchanged, and $K_{k}$ is updated by solving (\ref{CT LS}), $K_k$ converges to $K^*_{\alpha_k}$ as $k\rightarrow\infty$. That is, there exists $k'>k$ such that 
	$$||K_{k'}-K^*_{\alpha_k}||<\delta_1/2.$$
	When $\eta<\delta_2$, we have 
	$$||K_{k'}-K^*_{\alpha'}||\leq||K_{k'}-K^*_{\alpha_k}||+||K^*_{\alpha_k}-K^*_{\alpha'}||<\delta_1,$$
	By Lemma \ref{le delta1}, $A-\alpha' I_n-BK_{k'}$ is stable. 
	
	On the other hand, note that the solution $P_k$ to (\ref{CT Lyapunov}) is continuous on $\alpha_{k+1}$ and $K_k$, thus is uniformly continuous for $\alpha\in[0,\alpha_0]$ and $K\in\{K\in\mathbb{R}^{m\times n}: ||K-K^*_\alpha||<\delta_1\}$. As a result, there is $\delta_3>0$, such that if $|\alpha_{k+1}-\alpha_k|<\delta_3$, then $||P_k-P_{k-1}||<\sigma$.
	
	Consequently, once we choose $\eta<\min\{\delta_2,\delta_3\}$ and $S_0\in\mathbb{N}_+$ such that $\alpha_0=\eta S_0$ renders $A-\alpha_0 I_n$ stable, $\alpha_k$ will converge to $\alpha^*<\eta$ at some step $k''>0$, which implies that $\alpha^*=0$. Continue running Step 3 of Algorithm \ref{alg:CTARL}, $K_k$ converges to the optimal control gain $K^*$ as $k\rightarrow\infty$. Moreover, using any integer $S>S_0$ still works since $A-S\eta I_n$ is stable if $A-S_0\eta I_n$ is stable.
	\QEDA

\end{document}